\documentclass[final]{colt2026}

\usepackage{xcolor}
\definecolor{color1}{HTML}{105e8a}
\definecolor{color2}{HTML}{e99926}
\definecolor{color3}{HTML}{b82a0c}
\definecolor{color4}{HTML}{3e8a10}
\definecolor{color5}{HTML}{80037e}
\definecolor{color6}{HTML}{070707}

\title[Actively Learning Halfspaces]{Actively Learning Halfspaces without Synthetic Data}
\usepackage{times}

\usepackage{enumitem}
\usepackage[nameinlink]{cleveref}

\newcommand{\ds}{f}
\newcommand{\DS}{\mathrm{DS}}

\DeclareMathOperator*{\argmax}{arg\,max}
\DeclareMathOperator*{\argmin}{arg\,min}

\newtheorem{claim}[theorem]{Claim}
\newtheorem{observation}[theorem]{Observation}
\newtheorem{problem}[theorem]{Problem}

\newcommand{\ignore}[1]{}

\newcommand{\cA}{{\cal A}}

\newcommand{\cC}{{\cal C}}

\newcommand{\cE}{{\cal E}}

\newcommand{\cH}{{\cal H}}
\newcommand{\cI}{{\cal I}}

\newcommand{\cM}{{\cal M}}

\newcommand{\eps}{\varepsilon}

\newcommand{\NN}{\mathbb{N}}
\newcommand{\RR}{\mathbb{R}}

\newcommand{\norm}[1]{\left\lVert #1 \right\rVert}

\newcommand{\Sec}[1]{\hyperref[sec:#1]{\Cref*{sec:#1}}} 
\newcommand{\Eqn}[1]{\hyperref[eq:#1]{(\ref*{eq:#1})}} 
\newcommand{\Fig}[1]{\hyperref[fig:#1]{Fig.\,\ref*{fig:#1}}} 
\newcommand{\Tab}[1]{\hyperref[tab:#1]{Tab.\,\ref*{tab:#1}}} 
\newcommand{\Thm}[1]{\hyperref[thm:#1]{Theorem\,\ref*{thm:#1}}} 
\newcommand{\Fact}[1]{\hyperref[fact:#1]{Fact\,\ref*{fact:#1}}} 
\newcommand{\Lemma}[1]{\hyperref[#1]{Lemma\,\ref*{#1}}} 
\newcommand{\Prop}[1]{\hyperref[prop:#1]{Prop.~\ref*{prop:#1}}} 
\newcommand{\Cor}[1]{\hyperref[cor:#1]{Corollary~\ref*{cor:#1}}} 
\newcommand{\Conj}[1]{\hyperref[conj:#1]{Conjecture~\ref*{conj:#1}}} 
\newcommand{\Def}[1]{\hyperref[def:#1]{Definition~\ref*{def:#1}}} 
\newcommand{\Alg}[1]{\hyperref[alg:#1]{Alg.~\ref*{alg:#1}}} 
\newcommand{\Obs}[1]{\hyperref[#1]{Observation~\ref*{#1}}} 
\newcommand{\Ex}[1]{\hyperref[ex:#1]{Ex.~\ref*{ex:#1}}} 
\newcommand{\Clm}[1]{\hyperref[#1]{Claim~\ref*{#1}}} 
\newcommand{\Step}[1]{\hyperref[step:#1]{Step~\ref*{step:#1}}} 
\newcommand{\Problem}[1]{\hyperref[#1]{Problem~\ref*{#1}}} 
\newcommand{\Line}[1]{\hyperref[#1]{Line~\ref*{#1}}} 

\newenvironment{proofof}[1]{\noindent \textbf{Proof of #1.}}{}

\hypersetup{
    colorlinks,
    linkcolor={color1},
    citecolor={color1},
    urlcolor={color1}
}

\usepackage{algorithm}
\usepackage{algpseudocode}

\coltauthor{%
 \Name{Hadley Black} \Email{hadley.black@baruch.cuny.edu}\\
 \addr Baruch College, CUNY
 \AND
 \Name{Kasper Green Larsen} \Email{larsen@cs.au.dk}\\
 \addr Aarhus University
 \AND
 \Name{Arya Mazumdar} \Email{arya@ucsd.edu}\\
 \addr University of California, San Diego
 \AND
 \Name{Barna Saha} \Email{bsaha@ucsd.edu}\\
 \addr University of California, San Diego
 \AND
 \Name{Geelon So} \Email{geelon@ucsd.edu}\\
 \addr University of California, San Diego
}

\begin{document}
\maketitle

\begin{abstract}
    In the classic \emph{point location} problem, one is given an arbitrary dataset $X \subset \RR^d$ of $n$ points with \emph{query access} to an unknown halfspace $f \colon \RR^d \to \{0,1\}$, and the goal is to learn the label of every point in $X$. This problem is extremely well-studied and a nearly-optimal $\smash{\widetilde{O}(d \log n)}$ query algorithm is known due to Hopkins-Kane-Lovett-Mahajan (FOCS 2020). However, their algorithm is granted the power to query arbitrary points outside of $X$ (point synthesis), and in fact without this power there is an $\Omega(n)$ query lower bound due to Dasgupta (NeurIPS 2004). Nonetheless, query access to arbitrary synthesized data points is unrealistic in many contexts.

Our objective in this work is to design efficient algorithms for learning halfspaces \emph{without point synthesis}. To circumvent the $\Omega(n)$ lower bound, we consider learning halfspaces whose normal vectors come from a known set of size $D$, and show tight bounds of $\Theta(D + \log n)$. As a corollary, we obtain an optimal $O(d + \log n)$ query deterministic learner for the fundamental class of \emph{decision stumps} (depth-one decision trees, or axis-aligned halfspaces), closing a previous gap of $O(d \log n)$ vs. $\Omega(d + \log n)$ left open in the active learning literature. In fact, our algorithm solves the more general problem of learning a Boolean function $f$ over $n$ elements which is monotone under at least one of $D$ provided \emph{orderings} of these elements. Our technical insight is to exploit the structure in these orderings to essentially perform a binary search in parallel rather than considering each ordering sequentially, and we believe our approach may be of broader interest.

Furthermore, we use our exact learning algorithm to obtain nearly optimal algorithms for PAC-learning. We show that $O(\min(D + \log(1/\varepsilon), 1/\varepsilon) \cdot \log D)$ queries suffice to learn $f$ within error $\eps$, even in a setting when $f$ can be adversarially corrupted on a $c\varepsilon$-fraction of points, for a sufficiently small constant $c$. This bound is optimal up to a $\log D$ factor, including in the realizable setting.
\end{abstract}

\section{Introduction}

Halfspaces, or linear threshold functions (LTFs), are one of the most fundamental and well-studied objects in computer science and geometry, with numerous applications in machine learning, optimization, and beyond. The halfspace with normal vector $u \in S^{d-1}$ and threshold $t \in \RR$ bisects $\RR^d$ into two regions, classifying $x$ using the rule $\mathbf{1}(\langle u,x \rangle > t)$. In this work, we consider the following halfspace learning question. Given a dataset $X \subset \RR^d$ partitioned into two classes by an unknown halfspace $f \colon \RR^d \to \{0,1\}$, we are allowed to adaptively query the label $f(x)$ for any $x \in X$ and our goal is to learn the label of every point (or an all but $\eps$-fraction of points) in $X$. If an algorithm is allowed to query points outside of $X$, then we say it uses \emph{point synthesis}.  

With point synthesis, there is an information-theoretic lower bound of $\Omega(d \log n)$ queries and there has been extensive work on obtaining upper bounds \cite{DBLP:journals/jacm/Heide84,DBLP:journals/iandc/Meiser93,DBLP:conf/esa/CardinalIO16,DBLP:conf/icalp/KaneLM18,DBLP:journals/dcg/EzraS19} culminating\footnote{Many works refer to the dual problem of \emph{point location} where the goal is to determine the cell containing a hidden point in an arrangement of $n$ hyperplanes with the ability to query on which side of any given hyperplane the point lies. This is known to be equivalent to the halfspace learning question and so they are often referred to interchangeably.} in a $O(d \log^2 d \log n)$ query algorithm due to \cite{DBLP:conf/focs/HopkinsKLM20}. However, the ability to synthesize and query arbitrary data points is a strong assumption which may be infeasible in many settings, and our primary motivation in this work is to investigate when and how this can be avoided. The bad news is that, in general, without point synthesis no non-trivial algorithm is possible due to the following simple $\Omega(n)$ lower bound observed by \cite{DBLP:conf/nips/Dasgupta04} which holds even in two dimensions: for a set $X$ of $n$ distinct points on the unit circle, it is straightforward to design a family of halfspaces $\{f_x \colon x \in X\}$ such that $f_x$ uniquely assigns $f_x(x) = 1$ among all points in $X$, and thus one must query every point in order to learn every label\footnote{Note that this argument holds for \emph{any choice} of $n$ distinct points on the unit circle. For instance, this demonstrates an $\Omega(n)$ lower bound even for the class of halfspaces whose normal vector comes from an arbitrarily narrow cone.}.

\paragraph{Learning halfspaces optimally in bounded directions.} The above lower bound comes from the fact that there are an unbounded number of \emph{directions} the halfspace can be in, and each direction can be used to separate a single point in $X$ from the rest. Various ways of circumventing this lower bound have been considered in the literature (see \Cref{sec:related-work}); however, to do so without weakening the assumption of a worst-case point set or adding point synthesis, it is \emph{necessary} to bound the number of directions parameterizing the halfspace. It is not too hard to see that if $D \leq n$ directions are allowed, then the above argument gives an $\Omega(D)$ lower bound and $\Omega(\log n)$ queries are needed (to determine the threshold) even in one dimension, leading to an $\Omega(D + \log n)$ lower bound. On the other hand it is not too challenging to design an algorithm using $O(D \log n)$ queries by projecting $X$ along each allowed direction, and performing a binary search independently using the ordering on $X$ induced by each projection. This algorithm simply exploits the fact that the labels are monotone along the correct direction. In fact, even for the simple special case of \emph{axis-aligned} halfspaces (depth-one decision trees) in $\RR^d$, there is a gap of $O(d \log n)$ vs. $\Omega(d + \log n)$ that exists in the active learning literature \cite{DBLP:journals/jmlr/HannekeY15}. We investigate the core algorithmic question underlying this gap: \emph{can we exploit structure beyond monotonicity of the labels to go beyond a naive, independent application of binary search?} We give an affirmative answer, designing a new algorithm that achieves the optimal $O(D + \log n)$ query complexity, thus closing this gap.


\paragraph{A new ``parallel'' binary search learning algorithm.} In fact, we design an optimal algorithm for the following generalization of the above problem, which we believe may be of general interest: given a set $X$ of size $n$, query access to a function $f \colon X \to  \{0,1\}$, and $D$ permutations $\sigma_1,\ldots,\sigma_D \colon X \to [n]$ such that $f \circ \sigma_{i^{\star}}^{-1} \colon [n] \to \{0,1\}$ is \emph{monotone} for some $i^{\star} \in [D]$, we wish to learn $f(X)$ using few queries\footnote{The halfspace learning problem we consider 
corresponds to the special case of $X \subset \RR^d$ being a set of $n$ points, $f$ being a halfspace whose normal vector is from a known set $V = \{u_1,\ldots,u_D\}$ of unit vectors, and $\sigma_i$ being the ordering on $X$ by sorting according to the length of the projections $\langle x,u_i \rangle$ for each $x \in X$.}. A naive approach is to run an independent binary search on each $f \circ \sigma_i^{-1}$, ignoring the shared information across the $f \circ \sigma_i^{-1}$-s provided by a single query. (See \Cref{obs:simple} for a formal argument.) Exploiting this shared information requires carefully considering the structure of the permutations in order to choose queries to effectively perform these binary searches in parallel. We are able to use $O(1)$ queries to either reduce $D$ by one, or reduce $n$ by a constant factor, performing a step of binary search simultaneously on every $f \circ \sigma_i^{-1}$. At a high level, our algorithm and its analysis demonstrate how one can carefully choose queries to obtain significant speedups for analyzing multiple sequences simultaneously. We believe that our ideas for speeding up parallel binary search, and even our algorithm as a black-box, may be useful for obtaining similar speedups over the naive algorithm for other problems. 



\paragraph{Main application: decision stumps, or axis-aligned halfspaces.} A natural class of halfspaces with bounded directions is the class of \emph{axis-aligned} halfspaces. Such functions are equivalent to \emph{decision stumps} (depth-one decision trees), which classify data by thresholding with respect to a single feature. They are the basic building blocks of decision trees, commonly used in practice as base learners in ensemble learning techniques such as bagging and boosting~\cite{DBLP:journals/ml/Holte93,DBLP:conf/ecml/OliverH94,DBLP:conf/aaai/BanihashemHS23}. For learning decision stumps over arbitrary point sets in $\RR^d$, there is a gap of $O(d \log n)$ vs. $\Omega(d + \log n)$ left open in the active learning literature \cite{DBLP:journals/jmlr/HannekeY15}. As a direct corollary of our main result, we obtain a deterministic algorithm using $O(d + \log n)$ queries, closing the book on this fundamental concept class. 

Beyond depth one, there is a large body of research on decision tree learning over the Boolean hypercube, $\{0,1\}^d$ (e.g. \cite{DBLP:journals/jacm/BlancLQT22, DBLP:conf/innovations/BlancLT20,DBLP:conf/alt/BshoutyH19,DBLP:journals/iandc/EhrenfeuchtH89}); however, we note that there is an $\Omega(n)$ query lower bound for learning \emph{Boolean} decision trees of depth at least two\footnote{Let $X = \{(i,-i) \colon i \in [-n,n]\} \subset \RR^2$ and for each $i \in [-n,n]$ consider the depth two decision tree defined by $f_i(x) = \mathbf{1}(x_1 \geq i  \wedge x_2 \geq -i)$. Then, $f_i(i,-i) = 1$ and all other points in $X$ are labeled $0$. Thus, learning the labels exactly is equivalent to an unstructured search problem, implying we need $\Omega(n)$ queries.} when the domain $X \subset \RR^d$ is arbitrary. On the other hand, an alternative generalization is the problem of learning \emph{multi-class} decision trees, where each leaf must correspond to a distinct label, over arbitrary $X \subset \RR^d$. We do not know of any prior works studying this question and we believe this may be an interesting direction for future work. 

Beyond the axis-aligned setting, a more general class of halfspaces where our results have implications are \emph{sparse discrete} halfspaces: given a sparsity parameter $s \in \{1,\ldots, d\}$, suppose we restrict the normal vector to lie in $\{-1,0,1\}^d$ (or a larger hypergrid) with at most $s$ non-zero coordinates. The total number of directions for such halfspaces is bounded by $d^{O(s)}$ and so we obtain a learner with query complexity $d^{O(s)} + O(\log n)$ for exactly learning such a halfspace over an arbitrary point set $X \subset \RR^d$. This notion of sparse halfspaces has been considered in numerous works (e.g. \cite{DBLP:conf/approx/MatulefORS09}, \cite{DBLP:journals/tcs/AbasiAB16}, \cite{DBLP:journals/corr/abs-2502-16008} for a small sampling), but not for the problem we consider, to the best of our knowledge.





\paragraph{Learning preliminaries.} Besides exact learning, we design algorithms for approximate learning in both the realizable and tolerant settings. For two Boolean-valued functions $f,g \colon X \to \{0,1\}$ over the same domain, we use the notation $\norm{f-g}_1$ to denote their Hamming distance {\em normalized} by $|X|$. We consider learning a concept class $\cC \colon X \to \{0,1\}$ in the  settings defined below:
\begin{itemize}[leftmargin=*]
    \item \textbf{Exact learning:} Given $f \in \cC$, learn the label of every point in $X$.
    \item \textbf{Proper, realizable $(\eps,\delta)$-PAC learning:} Given $\eps > 0, \delta > 0$, and $f \in \cC$, learn $h \in \cC$ such that $\Pr_h[\norm{f - h}_1 > \eps] < \delta$. 
    \item \textbf{Proper, $c$-tolerant $(\eps,\delta)$-PAC learning:} Given $c > 0, \eps > 0, \delta > 0$, and $f$ such that $\norm{f - g}_1 \leq c\eps$ for some $g \in \cC$, learn $h \in \cC$ such that $\Pr_h[\norm{f - h}_1 > \eps] < \delta$.
\end{itemize}

Our algorithms for PAC-learning utilize our exact learning algorithm as a black-box, in conjunction with a novel randomized binary search procedure for tolerantly learning monotone functions, and a simple result from property testing of monotone functions.




\subsection{Results}

We now state our main results. Given $X \subset \RR^d$ and a set $V = \{u_1,\ldots,u_D\} \subset S^{d-1}$ of $D$ unit vectors, let $\cH(V,X)$ denote the set of all (non-homogeneous) halfspaces defined over $X$ with normal vector in $V$. That is, $f \in \cH(V,X)$ iff $f(x) = \mathbf{1}(\langle u, x \rangle  > t)$ for some $u \in V$ and $t \in \RR$. All of our algorithms run in polynomial time.\footnote{Since our focus is on query complexity we do not describe how to optimize the time complexity of our algorithms, nor do we provide a detailed runtime analysis.}

\begin{theorem} [$D$-Directional Halfspace Learning] \label{thm:D-direction-HS} Given arbitrary set $X \subset \RR^d$ of $n$ points and $V \subset S^{d-1}$ of $D$ unit vectors, the following learning algorithms exist for the concept class $\cH(V,X)$. 
\begin{enumerate}
    \item A deterministic exact learner using $O(\min(D + \log n,n))$ queries.
    \item A proper, realizable $(\eps,0)$-PAC learner using $O(\min(D + \log(\frac{1}{\eps}), \frac{\log D}{\eps}))$ queries.
    \item A proper $c$-tolerant $(\eps,\delta)$-PAC learner using $O(\min(D + \log(\frac{1}{\eps}),\frac{1}{\eps}) \cdot \log D)$ queries, for any constant $\delta > 0$, and sufficiently small constant $c \in (0,1)$.
\end{enumerate}    
\end{theorem}


\noindent \textbf{Decision stumps.} Next, we observe that \Cref{thm:D-direction-HS} implies improved bounds (in some cases optimal) for decision stumps as an immediate corollary. Given $i \in [d]$ and $t \in \RR$, the decision stump $\ds_{i,t} \colon\RR^d \to \{0,1\}$ is defined as $\ds_{i,t}(x) = \mathbf{1}(x_i \geq t)$. We denote the class of decision stumps over $\RR^d$ as $\DS_d$. The work of \cite{DBLP:journals/jmlr/HannekeY15} gives general upper and lower bounds on the query complexity of active learning in terms of a quantity which they call the \emph{star number} of a concept class. For learning decision stumps, their results imply lower and upper bounds of 
\[
\Omega\left(\min\Big(d + \log \frac{1}{\eps}, \frac{1}{\eps}\Big) + \min \Big(\log d, \log n\Big)\right) \text{ and } O\Big(\min\Big(d \log\Big(\frac{1}{\eps}\Big), \frac{\log d \log (1/\eps)}{\eps},n\Big)\Big)
\]
queries, respectively for $\eps$-realizable learning with constant success probability. Note that, replacing ``$d$" by ``$D$", this lower bound also holds for the $D$-directional halfspace learning problem considered in \Cref{thm:D-direction-HS}, since axis-aligned halfspaces are a special case. 
For completeness, we give a formal proof of these bounds in \Cref{sec:star-num}. For example, the previous known bounds for exact learning are $\Omega(\min(d + \log n),n)$ vs. $O(\min(d \log n,n))$. As an immediate corollary of \Cref{thm:D-direction-HS}, we obtain the following bounds for learning decision stumps which are tight in the exact learning setting and are within a $O(\log d)$ factor of the lower bound in the approximate learning settings.  

\begin{corollary} [Decision Stump Learning] \label{cor:decision-stump} Given an arbitrary set $X \subset \RR^d$ of $n$ points, the following learning algorithms exist for learning the class $\DS_d$ of decision stumps over $X$. 
\begin{enumerate}
    \item A deterministic exact learner using $O(\min(d + \log n,n))$ queries.
    \item A proper, realizable $(\eps,0)$-PAC learner using $O(\min(d + \log(\frac{1}{\eps}), \frac{\log d}{\eps}))$ queries.
    \item A proper, $c$-tolerant $(\eps,\delta)$-PAC learner using $O(\min(d + \log(\frac{1}{\eps}),\frac{1}{\eps}) \cdot \log d)$ queries, for any constant $\delta > 0$, and sufficiently small constant $c \in (0,1)$.
\end{enumerate}    
\end{corollary}

\subsection{Related Work} \label{sec:related-work}

We study active learning in the \emph{pool-based} setting: we are given access to a set $X \subset \mathbb{R}^d$ of $n$ unlabeled points and may adaptively query the labels of these points \cite{settles2009active}. For arbitrary non-homogeneous halfspaces in the exact, no-synthesis setting, there is a strong worst-case lower bound on the number of queries needed to label every point in $X$: even if the points lie on a circle in $\mathbb{R}^2$, the worst-case label complexity is $\Omega(n)$ queries \cite{DBLP:conf/nips/Dasgupta04}. 

There are various ways to circumvent this barrier. One way is to introduce additional assumptions about the label function and the data: assume that the halfspace goes through the origin and that the data is drawn from a nice distribution. For example, standard choices are the uniform distribution over the unit sphere or a log-concave distribution  \cite{DBLP:journals/jmlr/DasguptaKM09,DBLP:conf/colt/BalcanL13}. Margin-based approaches exploit additional structure near the decision boundary, such as margin separability, small boundary mass, or noise conditions relating label uncertainty to distance from the boundary \cite{DBLP:conf/colt/BalcanBZ07,DBLP:journals/jmlr/GonenSS13}. In a related online selective sampling setting, margin information can also be used to decide when to query labels and to obtain mistake bounds in terms of hinge or soft-margin losses \cite{DBLP:journals/jmlr/Cesa-BianchiGZ06a}. Finally, another common way to avoid the lower bound is to strengthen the query model, for example by allowing point synthesis or membership queries, where the learner may query labels of points outside the given pool \cite{DBLP:conf/focs/HopkinsKLM20,DBLP:conf/nips/DiakonikolasKM24}. The downside to membership queries is that, in practice, synthetic data may not have a sensible label \cite{baum1992query,DBLP:conf/nips/Dasgupta04}.

This work parametrizes the hardness of learning halfspaces in a new, purely combinatorial way: we restrict the class of halfspaces to structured subclasses with \emph{bounded directionality}. We assume that the normal vectors of the halfspaces are contained in a set of $D$ known directions. Importantly, this assumption does not restrict the data in any way, admitting arbitrary, worst-case point sets. We also do not assume access to the stronger membership query oracle. 

The family of decision stumps is a canonical example of halfspaces with bounded directionality. For exact learning of decision stumps, prior general bounds imply an $O(d \log n)$ upper bound and a $\Omega(d + \log n)$ lower bound \cite{DBLP:journals/jmlr/HannekeY15}; our $O(d + \log n)$ bound closes this gap. In spirit, our algorithm resembles classic active learning algorithms like CAL and $A^2$, which only query points that lead to certain progress \cite{DBLP:journals/ml/CohnAL94,DBLP:conf/icml/BalcanBL06}. However, our algorithm is more \emph{aggressive} in the sense that not all progress is equal \cite{DBLP:journals/tcs/Dasgupta11}. Our algorithm identifies and exploits more informative points to achieve optimal query complexity.

\vspace{11pt}\noindent \textbf{Organization.} We prove our main result, item (1) of \Cref{thm:D-direction-HS}, in \Cref{sec:exact}. This result is used as a black-box in the proof of items (2) and (3) of \Cref{thm:D-direction-HS}, which we prove in \Cref{sec:approximate}. 

\section{Exact Learning} \label{sec:exact}


In this section we prove item (1) of \Cref{thm:D-direction-HS}. That is, we design a deterministic exact learning algorithm for halfspaces whose normal vector belongs to a known set $V = \{u_1,\ldots,u_D\} \subset S^{d-1}$ of size $D$, with the optimal $O(\min(D + \log n,n))$ query complexity. In fact, we consider a generalized formulation of our problem. Given a set of $n$ points $X \subset \RR^d$, let $\sigma_1,\ldots,\sigma_D \colon X \to [n]$ denote the $D$ permutations on $X$ induced by ordering the projections of each $x \in X$ onto each $u \in V$.\footnote{Ties can be broken arbitrarily.} That is, $\sigma_i(x) = j$ iff $x$ has the $j$'th smallest projection, $\langle x, u_i\rangle$. Then, a labeling $f \colon X \to \{0,1\}$ is consistent with a halfspace $\mathbf{1}(\langle x,u_{i^{\star}} \rangle > t^{\star})$ for some $i^{\star} \in [D], t^{\star} \in \RR$ iff $f$ is \emph{non-decreasing} when $X$ is ordered according to $\sigma_{i^{\star}}$. That is, $f \circ \sigma_{i^{\star}}^{-1} \colon [n] \to \{0,1\}$ is a monotone function. Therefore, it suffices to consider the following generalized learning problem. 


\begin{problem} [Learning Shuffled Monotone Functions] \label{def:LSMF} Let $X$ be a set of $n$ elements with a Boolean labeling $\ell \colon X \to \{0,1\}$. Let $\sigma_1,\ldots,\sigma_D \colon X \to [n]$ be $D$ permutations with the guarantee that for some unknown $i^{\star} \in [D]$, the labeling $\ell \circ \sigma_{i^{\star}}^{-1} \colon [n] \to \{0,1\}$ is non-decreasing. The problem of learning shuffled monotone functions is to learn the labeling using as few queries to $\ell$ as possible. \end{problem}

By the argument preceding \Problem{def:LSMF}, any algorithm for learning shuffled monotone functions with $n$ points and $D$ permutations immediately implies an algorithm for learning $\cH(V,X)$ where $X \subset \RR^d$ is an arbitrary set of $n$ points and $V \subset S^{d-1}$ is an arbitrary set of $D$ unit vectors. 

We start by describing a simple algorithm that achieves $O(D \log n)$ query complexity and then move on to obtaining the optimal $O(D + \log n)$ query learner in \Cref{sec:LSMF}. We begin with a simple claim that gives a useful and general subroutine.

\begin{claim} \label{clm:candidates} Let $f \colon X \to \{0,1\}$ and let $C$ be a collection of ``candidate" labeling functions over $X$. There is an algorithm $\mathsf{Contender}(f,C,X)$ using at most $|C|$ label queries to $f$ to return a function $h \in C$ such that $h = f$ whenever $f \in C$.
\end{claim}

\begin{proof} Suppose there exists $x \in X$ such that $h_1(x) \neq h_2(x)$ for some $h_1, h_2 \in C$. Then by querying $f(x)$, we either observe that $h_1(x) \neq f(x)$, in which case we remove $h_1$ from $C$, or we observe $h_2(x) \neq f(x)$, in which case we remove $h_2$ from $C$. Otherwise, all $h \in C$ define the same labeling over $X$ and now we return this labeling. Observe that if $f \in C$, then it will never be removed by this process. Thus, this algorithm satisfies the guarantee of the claim. \end{proof}



Now, given $i \in [D]$ and $j \in [n]$, let $h^{(i)}_j \colon X \to \{0,1\}$ be defined by $h^{(i)}_j(x) = \mathbf{1}(\sigma_i(x) \geq j)$. In the shuffled monotone function learning problem (\Cref{def:LSMF}), we are guaranteed that there is some $i^{\ast} \in [D]$, $j^{\ast} \in [n]$ such that $h^{(i^{\ast})}_{j^{\ast}} = \ell$. We first observe that combining \Clm{clm:candidates} with a basic binary search yields a simple $O(D \log n)$ query algorithm.

\begin{observation} \label{obs:simple} There is a deterministic algorithm for learning shuffled monotone functions using $O(\min(D \log n,n))$ queries. \end{observation}

\begin{proof} If $n < D\log n$, we can simply query the label of every point. Otherwise, the algorithm is defined as follows. Initialize a candidate set $C \gets \emptyset$. For each $i \in [D]$, perform a standard binary search on $\ell \circ \sigma_i^{-1} \colon [n] \to \{0,1\}$, which returns an index $j_i \in \{0\} \cup [n]$ such that $\ell \circ \sigma_i^{-1}(j_i) = 0$ and $\ell \circ \sigma_i^{-1}(j_i + 1) = 1$. (Define $\ell \circ \sigma_i^{-1}(0) = 0$, $\ell \circ \sigma_i^{-1}(n+1) = 1$.) Add $h^{(i)}_{j_i}$ to $C$. Afterwards, run the $\mathsf{Contender}$ procedure of \Clm{clm:candidates} on $C$.

By our use of binary search we are guaranteed that $j_{i^{\star}}$ satisfies $\ell \circ \sigma_{i^{\star}}^{-1}(j) = \mathbf{1}(j > j_{i^{\star}})$. Then, the correctness of the algorithm follows by \Clm{clm:candidates}. The number of queries made in the search phase is $O(D \log n)$ and the $\mathsf{Contender}$ procedure of \Clm{clm:candidates} uses at most $D$ queries. \end{proof}

On the other hand, the lower bound of \Cref{thm:Hanneke-Yang} asserts that $\Omega(\min(D + \log n,n))$ are needed (set $\eps < 1/n$). Note that the algorithm of \Obs{obs:simple} treats each $\ell \circ \sigma_i^{-1} \colon [n] \to \{0,1\}$ separately, by independently running a binary search procedure on each. In particular, it ignores the fact that a single query reveals information in every $\ell \circ \sigma_i^{-1}$ simultaneously. Designing an algorithm which exploits this requires careful consideration of the structure of the permutations $\sigma_1,\ldots,\sigma_D$ to choose which points to query. We show that by carefully considering this structure we can obtain an algorithm making $O(\min(D + \log n,n))$ queries, thus closing this gap. 

At a high level, this query complexity comes from the fact that we are able to spend only $O(1)$ queries to either remove a permutation from consideration (thus reducing $D$ by one), or find a large set of points to safely remove from consideration (reducing $n$ by a constant factor), effectively performing a step of binary search simultaneously in every $\ell \circ \sigma_i^{-1}$. 




\subsection{Optimal Deterministic Algorithm} \label{sec:LSMF}

In this section we prove the following theorem, which immediately implies item (1) of \Cref{thm:D-direction-HS}.

\begin{theorem} \label{thm:LSMF} For any $D,n$, there is a deterministic algorithm (\Alg{LSMF}), $\mathsf{ShuffledMonotoneLearner}$, that exactly learns a shuffled monotone function using $O(D + \log n)$ label queries. \end{theorem} 

\begin{proof} Our goal is to, for every $i \in [D]$, either (1) determine that $i \neq i^{\star}$ by finding a ``decreasing pair" $x,y \in X$ such that $\ell(x) = 1, \ell(y) = 0$ and $\sigma_i(x) < \sigma_i(y)$ or (2) find a ``boundary pair" $x,y \in X$ such that $\ell(x) = 0$, $\ell(y) = 1$ and $\sigma_i(y) - \sigma_i(x) = 1$. If (2) holds for some $i$, then we can put the candidate hypothesis $h^{(i)}_{\sigma_i(y)}$ into a ``candidate" set $C$ and remove $i$ from consideration. Once (1) or (2) has been satisfied for every $i \in [D]$, we are guaranteed that $\ell \in C$ and we can then use the $\mathsf{Contender}$ procedure of \Clm{clm:candidates} using at most $D$ queries to recover the labeling, $\ell$.

Let $C \gets \emptyset$ denote the candidate set which is initially empty and let $S \gets [D]$ denote the current set of coordinates under consideration which is initially $[D]$. Let $Z \gets X$ denote the set of points under consideration which is initially $X$. Let $x^{\ast}, y^{\ast}$ denote the boundary points for $i^{\ast}$. That is, $\sigma_{i^{\ast}}(y^{\ast}) = \sigma_{i^{\ast}}(x^{\ast}) + 1$ and $\ell(x^{\ast}) = 0$, $\ell(y^{\ast}) = 1$. The following lemma gives the main subroutine.


\begin{lemma} \label{lem:main-sub-EX} Let $S \subseteq [D]$ such that $i^{\star} \in S$, and let $Z \subseteq X$ with $x^{\star},y^{\star} \in Z$. Let $n := |Z|$. Suppose we are given for each $i \in S$, an interval $W_i \subset [n]$ such that both (a) $\sigma_{i^{\star}}(x^{\star}),\sigma_{i^{\star}}(y^{\star}) \in W_{i^{\star}}$ and (b) $\min_{i \in S} |W_i| \leq n/3$. There is a procedure $\mathsf{RemoveOrReduce}(S,Z,(W_i \colon i \in S))$ that uses at most $3$ queries and returns one of the following:
\begin{enumerate}
    \item An index $i \in S$ such that $i \neq i^{\star}$.
    \item An index $i \in S$ and a pair $x,y \in Z$ such that $\sigma_i(y) = \sigma_i(x) + 1$ and $\ell(x) = 0$, $\ell(y) = 1$. 
    \item A set $Z' \subset Z$ of size $|Z'| \leq \frac{2n}{3}$ such that $x^{\star},y^{\star} \in Z'$. Moreover, in this case the procedure returns the labeling $\ell(Z \setminus Z')$ on all points outside of $Z'$. 
\end{enumerate}
If we remove any of the assumptions that $i^{\star} \in S$, $x^{\star},y^{\star} \in Z$, or $\sigma_{i^{\star}}(x^{\star}),\sigma_{i^{\star}}(y^{\star}) \in W_{i^{\star}}$, then the procedure still returns either $i \in S$ such that $i \neq i^{\star}$ or a set $Z' \subset Z$ of size $|Z'| \leq \frac{2n}{3}$. \end{lemma}

We defer the proof of \Lemma{lem:main-sub-EX} to \Cref{sec:main-sub-proof-EX} and proceed with the proof of \Cref{thm:LSMF}. The algorithm is defined in \Alg{LSMF}.

\begin{algorithm2e}[ht]
\caption{Deterministic Algorithm for Learning Shuffled Monotone Functions \label{alg:LSMF}} 
\DontPrintSemicolon
\LinesNumbered
\KwIn{A set $X$ of size $|X| = n$, permutations $\sigma_1,\ldots,\sigma_D \colon X \to [n]$, and query access to a labeling $\ell \colon X \to \{0,1\}$ such that $\ell \circ \sigma_{i^{\star}}^{-1}$ is non-decreasing for some $i^{\star} \in [D]$}
\KwOut{A full description of the labeling function  $\ell$}
\textbf{Initialize} $S \gets [D]$, $Z \gets X$, $C \gets \emptyset$\;
\textbf{For} every $i \in S$, \textbf{query} $\sigma^{-1}_i(1),\sigma_i^{-1}(n)$. If every such query returns the same label, then halt and assert that all points in $X$ have this label\; \label{alg:line:for}
\While{$n > 10$ and $S \neq \emptyset$ \label{alg:line:while}} {
    Choose $k \in S$ arbitrarily. Let $z_1 = \sigma_k^{-1}(1)$ and $z_4 = \sigma_k^{-1}(n)$. Choose $z_2,z_3$ such that $|\sigma_k(z_r) - \sigma_k(z_{r+1})| < n/3$ for all $r \in \{1,2,3\}$. \textbf{Query} $z_1,z_2,z_3,z_4$\; \label{alg:line:choose}
    \textbf{If} $\ell(z_1) = 1$ or $\ell(z_4) = 0$, then set $S = S \setminus \{k\}$, and continue the While-loop from \Line{alg:line:while}\; \label{alg:line:if}
    \textbf{Else} we have at least one of each label among $z_1,\ldots,z_4$. For each $i \in S$, let $a_i$ be the maximum index where a $0$-label is observed and $b_i$ the minimum index where a $1$-label is observed among $z_1,z_2,z_3,z_4$ in $\sigma_i$. If $a_i > b_i$ for some $i \in S$, then set $S = S \setminus \{i\}$ and continue the While-loop from \Line{alg:line:while}. Otherwise, set $W_i = [a_i,b_i]$ for every $i \in S$. Note that in this case $\ell \circ \sigma_i^{-1}$ is non-decreasing on $\{\sigma_i(z_1),\sigma_i(z_2),\sigma_i(z_3),\sigma_i(z_4)\}$ for every $i \in S$ and therefore $\min_{i \in S} |W_i| \leq b_k - a_k < n/3$\;
    
    Run the procedure $\mathsf{RemoveOrReduce}(S,Z,(W_i \colon i \in S))$ from \Lemma{lem:main-sub-EX}\; \label{alg:line:else}
    \textbf{If} the procedure returns $i \in S$ as in Case 1 of \Lemma{lem:main-sub-EX}, then set $S \gets S \setminus \{i\}$\; \label{alg:line:if2}
    \textbf{If} the procedure returns $i \in S$ and a pair $x,y \in Z$ as in Case 2 of \Lemma{lem:main-sub-EX}, then set $S \gets S \setminus \{i\}$ and $C \gets C \cup \{h^{(i)}_{\sigma_i(y)}\}$\; \label{alg:line:if3}
    \textbf{If} the procedure returns $Z'$ as in Case 3 of \Lemma{lem:main-sub-EX}, then set $Z \gets Z'$ and $n = |Z'|$\; \label{alg:line:if4}
}
\textbf{If} $S \neq \emptyset$, then we have $n \leq 10$. In this case, \textbf{query} the label of every point in $Z$. For each $i \in S$, if we find $x,y$ such that $\ell(x) = 0, \ell(y) = 1$ and $\sigma_i(y) = \sigma_i(x) + 1$, set $C \gets C \cup \{h^{(i)}_{\sigma_i(y)}\}$\;
Run the procedure $\mathsf{Contender}(C,X)$ from \Clm{clm:candidates} and \textbf{return} the hypothesis that it outputs\;
\end{algorithm2e}

\paragraph{Analysis.} If the labeling is constant, then clearly the algorithm is correct by \Line{alg:line:for} and uses at most $2D$ queries. Otherwise, we will assume there is at least one point with each label. In particular $x^{\star}$ is the max $0$-labeled point in $\sigma_{i^{\star}}$ and $y^{\star}$ is the min $1$-labeled point in $\sigma_{i^{\star}}$. The following is the key invariant of the algorithm that guarantees correctness.

\begin{claim} [\Alg{LSMF} invariant] \label{clm:invariant} Throughout the execution of \Alg{LSMF}, we have that if $i^{\star} \in S$, then $x^{\star},y^{\star} \in Z$, and otherwise we have $\ell = h^{i^{\star}}_{\sigma_{i^{\star}}(y^{\star})} \in C$. \end{claim}

\begin{proof} We prove the claim by induction on the size of $S$ and the size of $Z$. For the base case, when $Z = X$ and $S = [D]$, the claim clearly holds. Now, in \Line{alg:line:if}, \ref{alg:line:else}, \ref{alg:line:if2}, a coordinate $i$ is only removed from $S$ if it is known that $i \neq i^{\star}$ (the correctness of \Line{alg:line:if2} is by \Lemma{lem:main-sub-EX}). In step \Line{alg:line:if3}, when a coordinate $i$ is removed from $S$, a hypothesis \smash{${h_j}^{\!(i)}$} is placed into $C$ with the guarantee that if $i = i^{\star}$, then $\smash{{h_j}^{\!(i)}}  = \ell$ by \Lemma{lem:main-sub-EX} item (2). Thus, in all cases when the size of $S$ reduces, we maintain the invariant. When the size of $Z$ reduces in step \Line{alg:line:if4} it is with the guarantee that $x^{\star},y^{\star}$ remain in the set by \Lemma{lem:main-sub-EX} item (3). Therefore, the invariant is maintained in all cases. \end{proof}

By construction of the points $z_1,\ldots,z_4$ in \Line{alg:line:choose}, either the size of $S$ is reduced by one, or the intervals constructed in \Line{alg:line:else} satisfy $\min_{i \in S}|W_i| \leq |W_k| < n/3$. Moreover, the interval $W_{i^{\star}}$ is guaranteed to contain $\sigma_{i^{\star}}(x^{\star})$ and $\sigma_{i^{\star}}(y^{\star})$ and so the collection of intervals $(W_i \colon i \in S)$ satisfies the conditions for the subroutine \Lemma{lem:main-sub-EX}. We always make progress either by decreasing the size of $S$ by one, or by decreasing the size of $Z$ by $n/3$. Thus, the total number of queries is $O(D + \log n)$ and given \Clm{clm:invariant} and \Clm{clm:candidates}, the algorithm always returns the correct solution. \end{proof}

\subsection[Main Subroutine: Proof of Lemma \ref*{lem:main-sub-EX}]{Main Subroutine: Proof of \Lemma{lem:main-sub-EX}} \label{sec:main-sub-proof-EX}

\begin{proof} We begin with some notation. Let $t := |\{x \in Z \colon \ell(x) = 1\}|$ denote the Hamming weight of the labeling. We are given an interval $W_i$ for each $i \in S$ with the guarantee that $\min_{i \in S}|W_i| < n/3$ and $\sigma_{i^{\star}}(x^{\star}),\sigma_{i^{\star}}(y^{\star}) \in W_{i^{\star}}$. Note this implies $t \in W_{i^{\star}}$. We will use the notation ``$x < W_i$" or ``$x > W_i$" to mean that $x$ lies entirely below or above the interval $W_i$. 

Since $t \in W_{i^{\star}}$, we know that $\ell(x) = 0$ whenever $\sigma_{i^{\star}}(x) < W_{i^{\star}}$ and similarly $\ell(x) = 1$ whenever $\sigma_{i^{\ast}}(x) > W_{i^{\star}}$. Thus, if we ever discover some $i$ and $x$ where this does not hold, then we know $i \neq i^{\star}$ and we can safely return $i$ which satisfies the guarantee specified in item (1) of the lemma. With this in mind, we define the following subsets of $Z$:
\begin{itemize}
    \item Blue points: $B = \{x \in Z \colon \exists i \in S\text{, } \sigma_i(x) < W_i\}$.
    \item Red points: $R = \{x \in Z \colon \exists i \in S \text{, } \sigma_i(x) > W_i\}$.
    \item Purple points: $P = \{x \in Z \colon \sigma_i(x) \in W_i \text{, } \forall i \in S\}$.
\end{itemize}
The utility of blue and red points is given by the following claim.

\begin{claim} [Blue/red points] \label{clm:red-blue} (a) Suppose we query $x \in B$ and we observe that $\ell(x) = 1$. Then we can find a coordinate $i \in S$ satisfying item (1) of the lemma. (b) Similarly, suppose we query $x \in R$ and we observe that $\ell(x) = 0$. Then we can find a coordinate $i \in S$ satisfying item (1) of the lemma. \end{claim}

\begin{proof} Suppose $x \in B$ and we observe $\ell(x) = 1$. By definition of $B$, there exists $i \in S$ such that $\sigma_i(x) < W_i$. We are guaranteed that $i \neq i^{\ast}$, so that item (1) of the lemma is satisfied. Similarly, suppose $x \in R$ and we observe that $\ell(x) = 0$. Then by definition of $R$, there exists $i \in S$ such that $\sigma_i(x) > W_i$. We are guaranteed that $i \neq i^{\ast}$, so that item (1) of the lemma is satisfied. \end{proof}

We now continue with the proof of \Lemma{lem:main-sub-EX}.

\paragraph{Case 1:} Suppose there exists $x \in B \cap R$. Then we query this point to obtain $\ell(x)$. If $\ell(x) = 1$, then we satisfy item (a) of \Clm{clm:red-blue}. Otherwise, we satisfy item (b) of \Clm{clm:red-blue}. Either way, we are guaranteed to find some $i \neq i^{\ast}$ and so we satisfy item (1) of the lemma. This uses $1$ query.

\begin{figure}
    \hspace*{3cm}
    \includegraphics[scale=0.3]{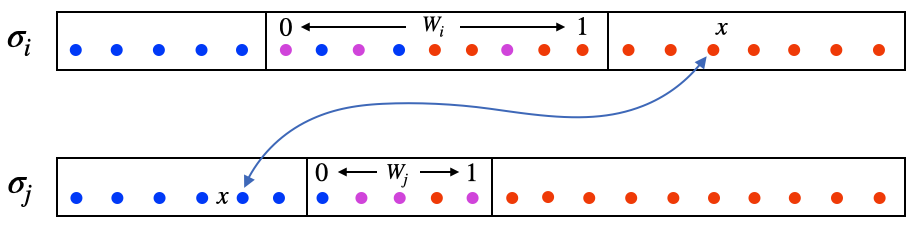}
    \caption{\small{Accompanying illustration for case 1 in the proof of \Lemma{lem:main-sub-EX}.} Here $x > W_i$ and $x < W_j$ and so by querying $\ell(x)$ we rule out one of $\sigma_i,\sigma_j$ as the monotone permutation. Thus, after case 1 we may assume $B \cap R = \emptyset$ and so the red, blue, purple coloring scheme is well-defined.}
    \label{fig:red-blue}
\end{figure}


(Refer to \Fig{red-blue}.) Thus, we may now assume that $B \cap R = \emptyset$, and in particular we have $Z = B \sqcup P \sqcup R$.  We now define the blue and red ``boundary points" of $\sigma_i$ as 
\[
\partial_{i,B} = \argmax_{x \in B} \sigma_i(x) \text{ and } \partial_{i,R} = \argmin_{x \in R} \sigma_i(x) \text{.}
\]

\paragraph{Case 2:} Suppose that for some $i \in S$ we have $\sigma_i(\partial_{i,R}) < \sigma_i(\partial_{i,B})$. I.e. there is a red point appearing before a blue point in $\sigma_i$. Then, we query $x:= \partial_{i,B}$ and $y:=\partial_{i,R}$. If $\ell(x) = 1$, then we satisfy item (a) of \Clm{clm:red-blue}. If $\ell(y) = 0$, then we satisfy item (b) of \Clm{clm:red-blue}. In either case we are guaranteed to find some $i' \neq i^{\ast}$ and so we satisfy item (1) of the lemma. Otherwise, we have $\ell(x) = 0$ and $\ell(y) = 1$ and since $\sigma_i(y) < \sigma_i(x)$ we know that $i \neq i^{\ast}$ and so we can simply return $i$ satisfying item (1) of the lemma. This case uses $2$ queries. (Refer to \Fig{red-blue-cross}.)

\begin{figure}
    \hspace*{3cm}
    \includegraphics[scale=0.3]{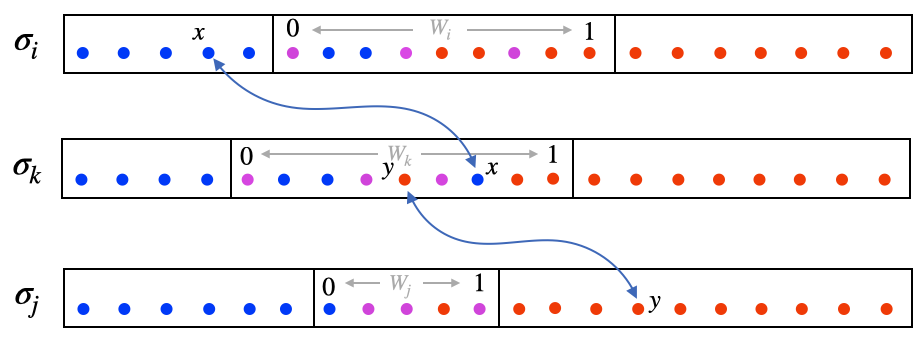}
    \caption{\small{Accompanying illustration for case 2 in the proof of \Lemma{lem:main-sub-EX}. Here, a red point $y$ appears before a blue point $x$ in $\sigma_k$. By definition of the colors, we have $x < W_i$ and $y > W_j$ for some $i,j$. By querying the labels of $x,y$ we rule out one of $\sigma_i,\sigma_j,\sigma_k$ as the monotone permutation. The picture for case 3 is analogous, except that $x$ appears just before $y$ in $\sigma_k$. Here, when $\ell(x) = 0, \ell(y) = 1$ we instead have found a ``boundary pair" for $\sigma_k$, satisfying item (2) of \Lemma{lem:main-sub-EX}.}}
    \label{fig:red-blue-cross}
\end{figure}


Thus, we may now assume that $\sigma_i(\partial_{i,B}) < \sigma_i(\partial_{i,R})$ for every $i \in S$. I.e. in every $\sigma_i$ all blue points appear before all red points.

\paragraph{Case 3:} Suppose that $P = \emptyset$. In this case we pick an arbitrary $i \in S$ and observe that $\sigma_i(\partial_{i,R}) = \sigma_i(\partial_{i,B}) + 1$. Thus, we query $x:=\partial_{i,B}$ and $y:=\partial_{i,R}$. If $\ell(x) = 1$ or $\ell(y) = 0$, then again we satisfy the premise of \Clm{clm:red-blue} and can return some $i' \in S$ satisfying item (1) of the lemma. Otherwise, we have $\ell(x) = 0$ and $\ell(y) = 1$ and thus we can return $i$ and $x,y$ which together satisfy item (2) of the lemma.  This case uses $2$ queries. 

Thus, we now assume that $P \neq \emptyset$. We now further partition $P$ into three sets as follows:
\begin{itemize}
    \item Green points: $G = \{x \in P \colon \exists i \in S\text{, } \sigma_i(x) < \sigma_i(\partial_{i,B})\}$.
    \item Orange points: $O = \{x \in P \colon \exists i \in S\text{, } \sigma_i(x) > \sigma_i(\partial_{i,R})\}$.
    \item Strong purple points: $P' = \{x \in P \colon \sigma_i(\partial_{i,B}) < \sigma_i(x) < \sigma_i(\partial_{i,R})\text{, } \forall i \in S\}$.
\end{itemize}
\begin{figure}
    \hspace*{3cm}
    \includegraphics[scale=0.3]{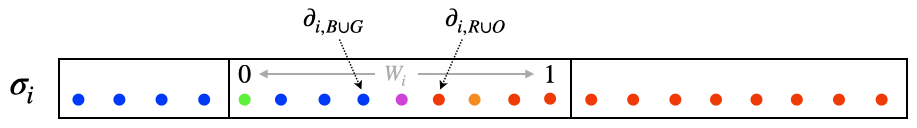}
    \caption{\small{Illustration for the blue, red, orange, green, purple coloring scheme for a single permutation. The left-most point in the interval $W_i$ was previously purple, and then re-colored green as it appears before a blue point. Similarly, the orange point was purple, but appears after a red point.}}
    \label{fig:orange-green}
\end{figure}

In words, green points are the ``non-blue" points which come before the blue boundary in \emph{some permutation} and orange points are the ``non-red" points that come after the red boundary in \emph{some permutation}. The remaining ``strong" purple points lie strictly between the blue and red boundaries in \emph{every permutation}. The utility of green and orange points is given by the following claim. 

\begin{claim} [Green/orange points] \label{clm:green-orange} (a) Suppose we query $x \in G$ and we observe that $\ell(x) = 1$. Then we can make at most one more query and find a coordinate $i \in S$ satisfying item (1) of the lemma. (b) Similarly, suppose we query $x \in O$ and we observe that $\ell(x) = 0$. Then we can make one more query and find a coordinate $i \in S$ satisfying item (1) of the lemma. \end{claim}

\begin{proof} Suppose $z \in G$ and we observe that $\ell(z) = 1$. By definition of $G$, there is some $x \in B$ such that $\sigma_i(z) < \sigma_i(x)$. Thus, we query this $x$ and obtain $\ell(x)$. If $\ell(x) = 0$, then we know $i \neq i^{\ast}$ and so we can return $i$ which satisfies item (1) of the lemma. Otherwise, if $\ell(x) = 1$, then by case (a) of \Clm{clm:red-blue} we can again find an $i' \in S$ satisfying item (1) of the lemma.

Suppose $z \in O$ and we observe that $\ell(z) = 0$. By definition of $O$, there is some $x \in R$ such that $\sigma_i(x) < \sigma_i(z)$. Thus, we query this $x$ and obtain $\ell(x)$. If $\ell(x) = 1$, then we know $i \neq i^{\ast}$ and so we can return $i$ which satisfies item (1) of the lemma. Otherwise, if $\ell(x) = 0$, then by case (b) of \Clm{clm:red-blue} we can again find an $i' \in S$ satisfying item (1) of the lemma. \end{proof}

We now continue with the proof of \Lemma{lem:main-sub-EX}.

\paragraph{Case 4:} Suppose that there exists $z \in G \cap O$. Then we query $z$ and obtain $\ell(z)$. If $\ell(z) = 1$, then we satisfy item (a) of \Clm{clm:green-orange}. Otherwise, if $\ell(z) = 0$, then we satisfy item (b) of \Clm{clm:green-orange}. In either case we find some $i \in S$ satisfying item (1) of the lemma. This case uses at most $2$ queries.

Thus, we now assume that $G \cap O = \emptyset$, and in particular we have $P = G \sqcup P'\sqcup O$.  We now define the blue-green and red-orange ``boundary points" of $\sigma_i$ as 
\[
\partial_{i,B \cup G} = \argmax_{x \in B \cup G} \sigma_i(x) \text{ and } \partial_{i,R \cup O} = \argmin_{x \in R \cup O} \sigma_i(x) \text{.}
\]

\paragraph{Case 5:} Suppose that for some $i \in S$ we have $\sigma_i(\partial_{i,R \cup O}) < \sigma_i(\partial_{i,B \cup G})$. Query $x := \partial_{i,B \cup G}$ and $y:=\partial_{i,R \cup O}$. If $\ell(x) = 1$, then we find $i' \in S$ satisfying item (1) of the lemma by item (a) of either \Clm{clm:red-blue} or \Clm{clm:green-orange}, depending on whether $x$ is blue or green. Similarly, if $\ell(y) = 0$, then we find $i' \in S$ satisfying item (1) of the lemma by item (b) of either \Clm{clm:red-blue} or \Clm{clm:green-orange}, depending on whether $y$ is red or orange. Otherwise, if $\ell(x) = 0$ and $\ell(y) = 1$, then we know $i \neq i^{\ast}$ and so we again satisfy item (1). This cases uses at most $3$ queries. (Refer to \Fig{orange-green-cross}.)

\begin{figure}
    \hspace*{3cm}
    \includegraphics[scale=0.3]{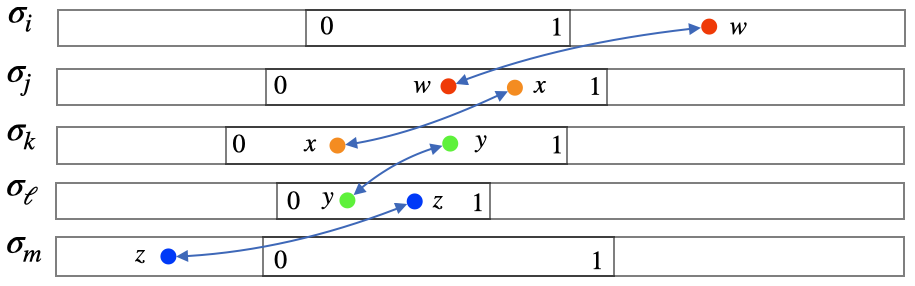}
    \caption{\small{Accompanying example illustration for case 5 in the proof of \Lemma{lem:main-sub-EX} in which an orange point $x$ appears before a green point $y$ in some $\sigma_k$. Since $y$ is green, $y$ appears before a blue point $z$ in some $\sigma_{\ell}$, which in turn means that $z < W_m$ for some $m \in [D]$. Similarly, since $x$ is orange, it appears after a red point $w$ in some $\sigma_{j}$, which means that $w > W_i$ for some $i \in [D]$. Querying these points reveals that one of $\sigma_i,\sigma_j,\sigma_k,\sigma_{\ell},\sigma_m$} does not have monotone labels. An accompanying illustration for case 6 is exactly the same, except that $y$ appears just before $x$ in $\sigma_k$.}
    \label{fig:orange-green-cross}
\end{figure}

Thus, we may now assume that $\sigma_i(\partial_{i,B \cup G}) < \sigma_i(\partial_{i,R \cup O})$. I.e. in every $\sigma_i$ all blue and green points appear before all red and orange points.

\paragraph{Case 6:} Suppose that $P' = \emptyset$. This case is similar to case 3. We pick an arbitrary $i \in S$ and observe that $\sigma_i(\partial_{i,R \cup O}) = \sigma_i(\partial_{i,B \cup G}) + 1$. Thus, we query $x := \partial_{i,B \cup G}$ and $y := \partial_{i,R \cup O}$ to obtain $\ell(x)$ and $\ell(y)$. If $\ell(x) = 1$ or $\ell(y) = 0$, then we again satisfy the premise of item (a) or (b) of either \Clm{clm:red-blue} or \Clm{clm:red-blue} and we can return some $i' \in S$ satisfying item (1) of the lemma. Otherwise, we have $\ell(x) = 0$ and $\ell(y) = 1$ and thus we can return $i$ and $x,y$ which together satisfy item (2) of the lemma. This case uses at most $3$ queries. 

Thus, we may now assume that $P' \neq \emptyset$. Recall that by definition of such points (``strong purple points"), any point $x \in P'$ separates all blue points from all red points in \emph{every permutation}. We will use this very strong property to satisfy item (3) of the lemma statement. Observe that in light of case 1, we know that $Z = B \sqcup P \sqcup R$ and so $|B| + |P| + |R| = n$. Moreover, by definition of $P$ and the bound $\min_{i \in S} |W_i| < n/3$ given by assumption in the lemma statement, we have $|P| \leq \min_{i \in S} |W_i| < n/3$. Thus, we are guaranteed that either $|B| \geq n/3$ or $|R| \geq n/3$.

\paragraph{Case 7:} Suppose $|B| \geq n/3$. In this case we wish to ``throw away" $B$ and let $Z' = P \cup R$. Recall the definition of $P'$ and the fact that $P' \neq \emptyset$ (in light of case 6). Pick an arbitrary $i \in S$ and let $z = \argmin_{z \in P'} \sigma_i(z)$ denote the left-most strong purple point for $\sigma_i$. Let $x$ be its left-neighbor, i.e. $\sigma_i(x) = \sigma_i(z)-1$. Observe that $x \in B \cup G$. We query $z$ and $x$ to obtain $\ell(z)$ and $\ell(x)$. If $\ell(x) = 1$, then we can satisfy item (1) of the lemma by either \Clm{clm:red-blue} or \Clm{clm:green-orange} depending on whether $x$ is blue or green. Thus, let's assume $\ell(x) = 0$. Then, if $\ell(z) = 1$, we can return $i$ with $x,z$ which together satisfy item (2) of the lemma. Otherwise, we have $\ell(z) = 0$. Now, since $z \in P'$, in every $\sigma_i$, all blue points appear before $z$. Moreover, in $\sigma_{i^{\ast}}$ we are guaranteed that the boundary points $x^{\ast}$ and $y^{\ast}$ appear at, or after, $z$. Putting these observations together, we know that $x^{\ast}, y^{\ast} \in P \cup R$. Thus, we return $Z' = P \cup R$ satisfying item (3) of the lemma statement. This case uses at most $3$ queries.

\begin{figure}
    \hspace*{3cm}
    \includegraphics[scale=0.3]{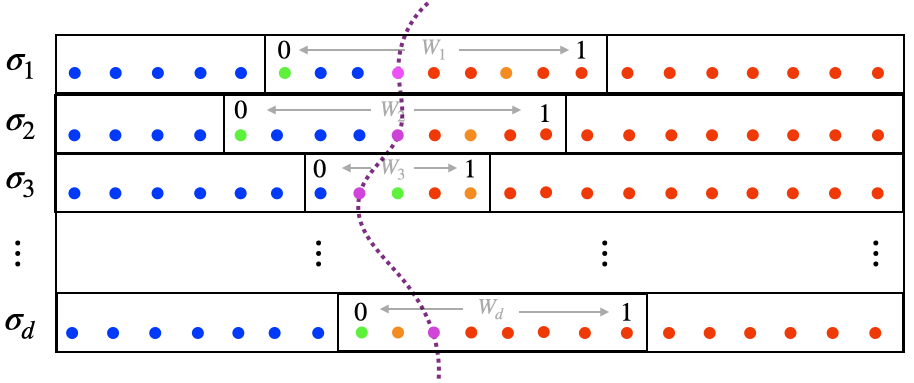}
    \caption{\small{Accompanying illustration for cases 7 and 8 in the proof of \Lemma{lem:main-sub-EX}. Here, we have found a single purple point which separates all blue points from all red points in every permutation simultaneously.}}
    \label{fig:purple-sep}
\end{figure}

\paragraph{Case 8:} Suppose $|R| \geq n/3$. In this case we wish to ``throw away" $R$ and let $Z' = P \cup B$. Recall the definition of $P'$ and the fact that $P' \neq \emptyset$ (in light of case 6). Pick an arbitrary $i \in S$ and let $z = \argmax_{z \in P'} \sigma_i(z)$ denote the right-most strong purple point for $\sigma_i$. Let $y$ be its right-neighbor, i.e. $\sigma_i(y) = \sigma_i(z)+1$. Observe that $y \in R \cup O$. We query $z$ and $x$. If $\ell(y) = 0$, then we can satisfy item (1) of the lemma by either \Clm{clm:red-blue} or \Clm{clm:green-orange} depending on whether $y$ is red or orange. Thus, let's assume $\ell(y) = 1$. Then, if $\ell(z) = 0$, we can return $i$ with $x,z$ which together satisfy item (2) of the lemma. Otherwise, we have $\ell(z) = 1$. Now, since $z \in P'$, in every $\sigma_i$, all red points appear after $z$. Moreover, in $\sigma_{i^{\ast}}$ we are guaranteed that the boundary points $x^{\ast}$ and $y^{\ast}$ appear at, or before, $z$. Putting these observations together, we know that $x^{\ast}, y^{\ast} \in P \cup B$. Thus, we return $Z' = P \cup B$ satisfying item (3) of the lemma statement. This case uses at most $3$ queries. (Refer to \Fig{purple-sep}.)

This completes the proof. \end{proof}

\bibliography{biblio}

\newpage
\appendix
\section{Approximate Learning} \label{sec:approximate}

In this section we consider $\eps$-learning a halfspace over an arbitrary set $X \subset \RR^d$ of $n$ points whose normal vector comes from a known set $V = \{u_1,\ldots,u_D\} \subset S^{d-1}$ of $D$ unit vectors. We prove items (3) and (2) of \Cref{thm:D-direction-HS} in \Cref{sec:agnostic} and \Cref{sec:realizable}, respectively. 

Our algorithms in this section crucially invoke our algorithm for exactly learning shuffled monotone functions obtained in \Cref{sec:LSMF}. This algorithm is combined with a randomized binary search procedure which gives a tolerant learner in the $1$-dimensional setting, along with a basic technique from property testing of monotone functions, and a careful probabilistic analysis.

Throughout the section, let $\cM$ denote the class of monotone Boolean functions over $[n]$, and for $f \colon [n] \to \{0,1\}$, let $d(f,\cM) = \min_{g \in \cM} \norm{f - g}_1$ denote the minimum Hamming distance of $f$ to any monotone function. Note that here $\norm{f - h}_1 = |\{j \in [n] \colon f(j) \neq h(j)\}| \cdot n^{-1}$. We use the notation $h_S$ to denote the restriction of $h$ to a subset $S$ of its domain.

\subsection{Tolerant Learning} \label{sec:agnostic}

The following lemma gives an optimal $c$-tolerant algorithm for learning monotone functions in $1$-dimension, for sufficiently small constant $c$. We use this algorithm as a subroutine for our $\eps$-learners in both the tolerant and realizable settings. 

\begin{lemma} [Randomized Binary Search] \label{lem:mono-agnostic-learn}  Let $f \colon [n] \to \{0,1\}$ be a function with $d(f,\cM) \leq \eps$. Then, there is a randomized algorithm $\cA$ using $O(\log (1/\eps)+\log(1/\delta))$ queries to $f$ that returns a monotone function $g \colon [n] \to \{0,1\}$ such that $\norm{f - g}_1 \leq 10\eps/\delta$ with probability at least $1-\delta$. \end{lemma}

We defer the proof of \Lemma{lem:mono-agnostic-learn} to \Cref{sec:mono-agnostic-learn}. Recall that $\cH(V,X)$ denotes the set of all (non-homogeneous) halfspaces defined over $X$ with normal vector in $V$. 

\begin{theorem} \label{thm:agnostic-formal} Let $X \subset \RR^d$ be any set of $n$ points and $V \subset S^{d-1}$ an arbitrary set of $D$ unit vectors. For any $\eps,\delta > 0$, suppose that $f : X \to \{0,1\}$ satisfies $\norm{f-g}_{1} \leq \frac{\eps\delta}{400}$ for some $g \in \cH(V,X)$. There is a randomized algorithm (\Alg{agnostic-d}) using $O(\ln (D/\delta) \cdot \min(D + \log(\frac{1}{\eps\delta}),\frac{1}{\eps\delta}))$ queries to $f$ and returns a halfspace $h \in \cH(V,X)$ such that $\norm{f-h}_1 \leq \eps$ with probability $1-2\delta$. \end{theorem}

\begin{proof} Our algorithm is defined in \Alg{agnostic-d}. Let $\sigma_1,\ldots,\sigma_D \colon X \to [n]$ denote the $D$ permutations on $X$ induced by ordering the projections of each $x \in X$ onto each $u \in V$. That is, $\sigma_i(x) = j$ iff $x$ has the $j$'th smallest projection, $\langle x, u_i\rangle$. First, since $\norm{f - g}_1 \leq \frac{\eps\delta}{400}$ for some $g \in \cH(V,X)$, it follows that $d(f\circ \sigma_{i^{\star}}^{-1},\cM) \leq \frac{\eps\delta}{400}$ for some $i^{\star} \in [D]$. The following claim (whose proof we defer momentarily) is the main technical component of the proof of correctness.

\begin{algorithm2e}[ht]
\caption{Tolerant Learner for Decision Stumps \label{alg:agnostic-d}} 
\DontPrintSemicolon
\LinesNumbered
\KwIn{A set $X \subset \RR^d$ of size $|X| = n$, a set $V = \{u_1,\ldots,u_D\} \subset S^{d-1}$ of unit vectors, and query access to $f \colon X \to \{0,1\}$ such that $\norm{f-g}_{1} \leq \frac{\eps\delta}{400}$ for some $g \in \cH(V,X)$}
\KwOut{A halfspace $h \in \cH(V,X)$ such that $\norm{f-h}_1 \leq \eps$ with probability $1-\delta$}
For $i \in [D]$, let $\sigma_i \colon X \to [n]$ be the ordering on $X$ induced by the projection of $X$ onto $u_i$\;
\For{$j \leq 50 \ln (D/\delta)$} {
Draw a set $S_j \subset X$ of $s := \frac{40}{\eps\delta}$ iid uniform samples from $X$\;
Run the exact learning algorithm $\mathsf{ShuffledMonotoneLearner}(f_{S_j},S_j,\sigma_1,\ldots,\sigma_D)$ of \Cref{thm:LSMF} on the restriction $f_{S_j}$ using $O(\min(D+\log(\frac{1}{\eps\delta}),\frac{1}{\eps\delta}))$ queries and let $f_{S_j}' \colon {S_j} \to \{0,1\}$ denote the labeling of $S_j$ it returns\;
For each $i \in [D]$, if $f_{S_j}' \circ \sigma_D$ is non-decreasing on $S_j$, set $Z_{ij} = 1$. Else, set $Z_{ij} = -1$\; \label{alg:line:Zij}
}
For each $i \in [D]$, let $Z_i = \sum_j Z_{ij}$ and set $i_R \gets \argmax_i Z_i$\; \label{alg:line:iR}
Run $\mathsf{RandomBinarySearch}(f \circ \sigma_{i_R}^{-1},\eps\delta/10,\delta)$, and output the labeling function it returns\; \label{alg:line:binarysearch}
\end{algorithm2e}

\begin{claim} \label{clm:concentration} With probability $1-\delta$, the index $i_R$ obtained in \Line{alg:line:iR} satisfies $d(f\circ \sigma_{i_R}^{-1},\cM) \leq \eps\delta/10$. \end{claim}

By \Clm{clm:concentration}, the call to $\mathsf{RandomBinarySearch}$ in \Line{alg:line:binarysearch} successfully returns a monotone function $h \colon [n] \to \{0,1\}$ such that $\norm{f - h \circ \sigma_{i_R}}_1 \leq \eps$ with probability $1-\delta$ by \Lemma{lem:mono-agnostic-learn}. By a union bound, the algorithm succeeds with probability $1-2\delta$. Moreover, since $h$ is monotone, we have $h \circ \sigma_{i_R} \in \cH(V,X)$ and so the algorithm is a proper learner. This completes the proof. \end{proof}


\begin{proofof}{\Clm{clm:concentration}} Let $F = \{i \in [D] \colon d(f\circ \sigma_{i}^{-1},\cM) > \eps\delta/10\}$ and observe that 
\begin{align} \label{eq:union-F}
    \Pr[d(f\circ \sigma_{i_R}^{-1},\cM) > \eps\delta/10] \leq \Pr[\exists i \in F \colon Z_i \geq Z_{i^{\star}}] \leq \sum_{i \in F} \Pr[Z_i - Z_{i^{\star}} \geq 0]
\end{align}
where $Z_i = \sum_j Z_{ij}$ is the random variable defined in \Line{alg:line:iR}. Note that $Z_i - Z_{i^{\star}} = \sum_j Y_j$ where $Y_j := Z_{ij} - Z_{i^{\star}j}$ are iid random variables supported on $\{-2,0,2\}$. We can now bound $\mathbb{E}[Y_j]$ and then bound $\Pr[Z_i-Z_{i^{\star}}]$ by standard concentration bounds and independence of the $Y_j$'s.

Let $\cE_{ij}$ denote the event that $f_{S_j} \circ \sigma_{i}^{-1}$ is non-decreasing. Now, when $\cE_{i^{\star}j}$ holds, the restricted labeling $f_{S_j}$ satisfies the pre-condition for the exact learning algorithm $\mathsf{ShuffledMonotoneLearner}$ of \Cref{thm:LSMF} to successfully compute the labeling. I.e., $\cE_{i^{\star}j}$ implies $f_{S_j}' = f_{S_j}$ in \Line{alg:line:Zij}. Note that when $\cE_{i^{\star}j}$ does not hold, we have no guarantee about $f_{S_j}'$. Nonetheless, by a union bound, we have
\begin{align} \label{eq:Y}
    \Pr[Y_j \neq -2] &\leq \Pr[Z_{ij} = 1] + \Pr[Z_{i^{\star}j} = -1] \nonumber \\ 
    &\leq \Pr[\neg\cE_{i^{\star}j} \vee \cE_{ij}] + \Pr[\neg \cE_{i^{\star}j}] \leq 2\Pr[\neg \cE_{i^{\star}j}] + \Pr[\cE_{ij}]
\end{align}

We now lower bound the probability of $\cE_{i^{\star}j}$. Let $h \colon [n] \to \{0,1\}$ be a monotone function such that $\norm{f \circ \sigma_{i^{\star}}^{-1} - h}_1 \leq \frac{\eps \delta n}{400}$ and observe that $|\Delta(f,h)| = |\{x \in X \colon f(x) \neq h \circ \sigma_{i^{\star}}(x)| < \frac{\eps\delta n}{400}$. Moreover, if $S_j \cap \Delta(f,h) = \emptyset$, then $\cE_{i^{\star}j}$ holds. Thus,
\begin{align} \label{eq:istar}
    \Pr_{S_j}[\cE_{i^{\star}j}] \geq \Pr_{S_j}[S_j \cap \Delta(f,h) = \emptyset] \geq \left(1-\frac{\eps\delta}{400}\right)^{\frac{40}{\eps\delta}} \geq e^{-1/10} > 9/10 \text{.}
\end{align}

We now upper bound $\Pr[\cE_{ij}]$ for $i \in F$. We will use the following lemma from the literature on property testing of monotone functions.

\begin{lemma} [Lemma 16, \cite{DodisGLRRS99}] \label{lem:mono-test} Let $g \colon [n] \to \{0,1\}$ such that $d(f,\cM) \geq \eps$. Then, choosing $s \geq 2$ points $j_1,\ldots,j_s \in [n]$ iid and uniformly at random from $[n]$, we observe a violation of monotonicity for $g$ with probability at least $(1-e^{-\eps s /2})^2$. \end{lemma}

Using \Lemma{lem:mono-test} and the definition of $F$, we have
\begin{align} \label{eq:ifar}
    i \in F ~\implies~ \Pr[\cE_{ij}] \leq 1 - \left(1-e^{-(\frac{\eps\delta}{10} \cdot \frac{40}{\eps\delta})/2}\right)^2 = e^{-2} < 1/7     
\end{align}
and plugging \cref{eq:istar} and \cref{eq:ifar} into \cref{eq:Y} yields $\Pr[Y_j \neq -2] < \frac{1}{5} + \frac{1}{7} < \frac{2}{5}$, which implies
\begin{align}
    \mathbb{E}[Y_j] \leq -2 \Pr[Y_j = -2] + 2\Pr[Y_j \neq -2] < -2 \cdot \frac{3}{5} + 2 \cdot \frac{2}{5} = -\frac{2}{5}
\end{align}
Therefore, recalling that $Z_i - Z_{i^{\star}} = \sum_{j \leq t} Y_j$, we have 
\[\mathbb{E}[Z_i - Z_{i^{\star}}] < -(2/5) \cdot 50\ln(D/\delta) = -20\ln(D/\delta)\]
and since the $Y_j$'s are independent, we have by Hoeffding's inequality that
\begin{align*}
\Pr\big[Z_{i} -  Z_{i^{\star}} \geq 0\big] &\leq \Pr\left[\big|(Z_{i} -  Z_{i^{\star}}) - \mathbb{E}[Z_{i} -  Z_{i^{\star}}]\big| > 20\ln(D/\delta)\right] 
\\&\leq \exp\left(- \frac{2 \cdot 400 \ln^2(D/\delta)}{16 \cdot 50 \ln(D/\delta)}\right) = \frac{\delta}{D}
\end{align*}
and plugging this bound into \cref{eq:union-F} completes the proof. \end{proofof}

\subsection{Realizable Learning} \label{sec:realizable}

In this section we obtain a slightly improved algorithm for the realizable setting by adapting the exact learning algorithm for shuffled monotone functions in \Cref{thm:LSMF}.


\begin{theorem} \label{thm:realizable-formal} Let $X \subset \RR^d$ be any set of $n$ points and $V \subset S^{d-1}$ an arbitrary set of $D$ unit vectors. For any $\eps,\delta > 0$, there is a randomized algorithm which, given query access to $f \colon X \to \{0,1\}$ such that $f \in \cH(V,X)$, uses $O(\min(D + \log(\frac{1}{\eps}),\frac{\ln (D/\delta)}{\delta\eps}))$ queries and returns a halfspace $h \in \cH(V,X)$ such that $\norm{f-h}_1 \leq \eps$ with probability $1-\delta$. If $D + \log(\frac{1}{\eps}) \leq \frac{\ln D}{\delta\eps}$, the algorithm is deterministic. \end{theorem}


\begin{proof} There is a $O(\frac{\ln (D/\delta)}{\delta\eps})$ query algorithm using the tolerant algorithm of \Cref{thm:agnostic-formal}. 

We show that in the realizable setting when $D + \log(1/\eps)) \ll \frac{\ln (D/\delta)}{\delta\eps}$, we can remove the $\ln(D/\delta)$ factor from the query complexity in \Cref{thm:agnostic-formal}. To achieve this, we describe how to slightly modify the exact learning algorithm of \Cref{thm:LSMF} to obtain a deterministic algorithm that correctly labels all but an $\eps$-fraction of $X$ using $O(D + \log(1/\eps))$ queries.

\begin{theorem} \label{thm:LSMF-eps} For any $D,n$ and $\eps > 0$, there is a deterministic algorithm that $\eps$-learns a shuffled monotone function using $O(D + \log(1/\eps))$ label queries, and returns a monotone labeling function. \end{theorem}

Again, since \Cref{thm:LSMF-eps} returns a monotone hypothesis, we obtain a proper learner. \end{proof}



\begin{proofof}{\Cref{thm:LSMF-eps}} We consider a slight modification of the algorithm described in \Alg{LSMF}. Change \Line{alg:line:while} so that the while-loop terminates when $n < \eps n$ (instead of $n < 10$). That is, we run the algorithm of \Cref{thm:LSMF} until there are $|Z| \leq \eps n$ points remaining. This uses $O(D + \log(1/\eps))$ queries and returns a set $C$ of candidate hypotheses which by \Clm{clm:invariant} is guaranteed to contain the true labeling $\ell$ in the case that $i^{\star} \notin S$. Now, by \Clm{clm:candidates} we can use at most $D$ queries and return a single representative $h^C \in C$ from this collection with the guarantee that if $\ell \in C$, then $h^C = \ell$. In summary, if $i^{\star} \notin S$, then we have $h^C = \ell$.

Next, we define at candidate hypothesis $h^Z \colon X \to \{0,1\}$ using the remaining points, $Z$. By item (3) of \Lemma{lem:main-sub-EX}, whenever we remove points from $Z$, we in fact learn the true label of these points and so we can simply define $h^{Z}(x) = \ell(x)$ whenever a point is removed. Finally, set $h^Z(Z) = 0$. Now, note that if $i^{\star} \in S$, then by \Clm{clm:invariant}, we have $x^{\star},y^{\star} \in Z$ and so $h^Z(x) = \ell(x)$ for every $x \in X \setminus Z$, and thus $\norm{f - h^Z} \leq \eps$. In summary, if $i^{\star} \in S$, then $\norm{f - h^Z} \leq \eps$. 

The only remaining task is to determine which of $h^C$ or $h^Z$ to return. If there exists a point $x \in X \setminus Z$ such that $h^C(x) \neq h^Z(x)$, then query $\ell(x)$ and return the hypothesis which agrees with $\ell$ at $x$. Otherwise, return $h^C$. If $i^{\star} \not\in S$, then $h^C = \ell$, as argued above. In this case we will always return $h^C$. Otherwise, as argued, we have $\norm{f - h^Z} \leq \eps$. If $h^Z(X \setminus Z) = h^C(X \setminus Z)$, then we also have $\norm{f - h^C} \leq \eps$ and this is what we return. Otherwise, we return $h^Z$. The returned hypothesis is $\eps$-close to $f$ in every case. Moreover, in both cases the returned function is monotone as claimed. \end{proofof}

\subsection{Randomized Binary Search} \label{sec:mono-agnostic-learn}

\begin{proofof}{\Lemma{lem:mono-agnostic-learn}} The algorithm (\Alg{random-binary-search}) performs a randomized binary search until there are fewer than $C \eps n$ unknown elements remaining (we will set $C \approx 1/\delta$ later on in the proof).

\begin{algorithm2e}[ht]
\caption{Randomized Binary Search Learner \label{alg:random-binary-search}} 
\DontPrintSemicolon
\LinesNumbered
\KwIn{$\eps,\delta > 0$ and query access to $f \colon [n] \to \{0,1\}$ such that $\norm{f-h}_{1} \leq \eps$ for some monotone function $h \colon [n] \to \{0,1\}$}
\KwOut{Monotone function $g \colon [n] \to \{0,1\}$ such that $\norm{f - g}_1 \leq 10\eps/\delta$ with probability $1-\delta$} 
Let $\cI \gets [n]$\;
\While{$|\cI| > C \eps n$} {
Partition $\cI$ into three subintervals of equal size $\cI = \cI_L \sqcup \cI_M \sqcup \cI_R$\;
Choose $i$ uniformly at random from $\cI_M$\footnote{This ensures every step reduces $|\cI|$ by a constant factor.} and query $f(i)$\;
If $f(i) = 0$, then set $g(\cI_L) = 0$ for all $z \leq i$ and set $\cI \gets \cI_M \cup \cI_R$\;
If $f(i) = 1$, then set $g(R) = 1$ for all $z \geq i$ and set $\cI \gets \cI_L \cup \cI_M$\;
}
For each $z \in \cI$, set $g(z) = 0$\;
\end{algorithm2e}

Let $h \in \cM$ minimize $\norm{f - h}_1$ and let $\Delta(f,h) = \{i \in [n] \colon f(i) \neq h(i)\}$. Note that $|\Delta(f,h)| \leq \eps n$. Observe that if every query of the above algorithm is from $[n] \setminus \Delta(f,h)$, then the output $g$ is guaranteed to satisfy $\norm{h - g}_1 \leq C\eps$ and therefore by the triangle inequality 
\[
\norm{f - g}_1 \leq \norm{f - h}_1 + \norm{h - g}_1 \leq (C + 1) \eps \text{.}
\]
Thus, it suffices to show that with high probability every query misses $\Delta(f,h)$. 

We first bound the number of queries $q$ made by the algorithm. Let $\cI_t$ denote the interval $\cI$ after the first $t-1$ queries. Note that by design of the algorithm, we always have $|\cI_t| = (2/3)|\cI_{t-1}|$ and so $|\cI_t| = (2/3)^t n$. Therefore, the total number of queries is bounded by
\[
q \leq \log_{3/2} n - \log_{3/2} (C \eps n) = \log_{3/2} (1/C\eps) \text{.}
\] 
Now, let $B_t$ denote the bad event that the $t$-th query lands in $\Delta(f,h)$ and let $B$ denote the union of these events for $t \leq q$. Since we choose the $t$-th query in step (2b) from the middle third of the interval $\cI_t$, we have $\Pr[B_t] \leq \frac{\eps n}{(1/3)|\cI_t|}$. Then, using our bounds on $\Delta(f,h)$, $q$, $|\cI_t|$, and a union bound we have
\begin{align*} 
    \Pr[B] \leq \sum_{t \leq q} \frac{3 \eps n }{(2/3)^t n} = 3\eps \sum_{t \leq q} (3/2)^t \leq 3\eps \cdot 3 (3/2)^q \leq 9/C
\end{align*}
and setting $C = 9/\delta$ completes the proof. \end{proofof}
\section{Lower Bound for Learning Decision Stumps} \label{sec:star-num}

\begin{theorem} [Theorem 3 of \cite{DBLP:journals/jmlr/HannekeY15}] \label{thm:Hanneke-Yang} 
Fix $n,d \in \NN$ and $\eps \geq \frac{1}{2n}$. Any membership query algorithm for $\eps$-realizable learning of decision stumps in $\RR^d$ over $n$ points must use
\[
\Omega\left(\min\Big(d + \log \frac{1}{\eps}, \frac{1}{\eps}\Big) + \min \Big(\log d, \log n\Big)\right) \textrm{ queries}
\]
in the worst case, and there is an algorithm using $O(\min(d \log(1/\eps), \frac{\log d \log (1/\eps)}{\eps}, n)$ queries. \end{theorem}

\begin{proof}
    Note that it is without loss of generality that we require $\eps \geq \frac{1}{2n}$. When $\eps < \frac{1}{n}$, it is the case that $\eps$-learning is equivalent to exact learning.
    
    Let $\Lambda(\cH)$ be the query complexity of any active learning algorithm for a concept class $\cH$. Theorem 3 of \cite{DBLP:journals/jmlr/HannekeY15} lower bounds $\Lambda(\cH)$ in the $\eps$-realizable setting with respect to three quantities: the star number $\mathfrak{s}(\cH)$, the VC dimension $\mathrm{VC}(\cH)$, and the cardinality $|\cH|$,
    \[\Lambda(\cH) \gtrsim \max\left\{\min\Big(\mathfrak{s}(\cH), \frac{1}{\eps}\Big), \mathrm{VC}(\cH),  \log \left(\min \Big(\frac{1}{\eps}, |\cH|\Big)\right)\right\}.\]
    The same theorem also provides the upper bound:
    \[\Lambda(\cH) \lesssim \min\left\{\mathfrak{s}(\cH), \frac{\mathrm{VC}(\cH)}{\eps}, \frac{\mathfrak{s}(\cH) \cdot \mathrm{VC}(\cH)}{\log \mathfrak{s}(\cH)}\right\} \log \frac{1}{\eps}.\]
    
    In the following, we let $X = \{x_1,\ldots, x_n\}$ be a finite point set in $\RR^d$, and let $\cH$ contain the set of all possible distinct labelings of $X$ by the class of decision stumps $\DS_d$,
    \[\cH = \big\{(f(x_1),\ldots, f(x_n) : f \in \DS_d\big\} \subset \{0,1\}^X.\]
    We can select $X$ so that the following complexity measures of $\cH$ are attained:
    \begin{itemize}
        \item The star number of decision stumps $\DS_d$ is $2d$, by \Cref{lem:ds-star-number}. This implies the upper bound $\mathfrak{s}(\cH) \leq 2d$, and directly from the definition of the star number of $\DS_d$, there is also a set $X$ such that $\mathfrak{s}(\cH) = \min\{2d,n\}$. 
        \item The VC dimension of $\DS_d$ is $\Theta(\log d)$ by \cite{DBLP:journals/tnn/Yildiz15}. This implies the upper bound $\mathrm{VC}(\cH) \leq \mathrm{VC}(\DS_d)$. Let $X$ be a shattered set of $\DS_d$, so that $\mathrm{VC}(\cH) = \min\{\log d, \log n\}$.
        \item Let $X$ be any point set such that the first coordinate of each point is distinct. Then, $|\cH| \geq n$.
    \end{itemize}
    By assumption, have $\min\big\{n, \frac{1}{\eps}\big\} = \frac{1}{\eps}$. This implies the following lower bound:
    \begin{align*} 
        \Lambda(\cH) &\geq \Omega\left(\min\Big(d, \frac{1}{\eps}\Big) + \min\Big(\log d, \log n\Big) + \min\Big(\log\frac{1}{\eps}, \log n\Big)\right)
        \\&= \Omega\left(\min\Big(d + \log \frac{1}{\eps}, \frac{1}{\eps}\Big) + \min \Big(\log d, \log n\Big)\right).
    \end{align*}
\end{proof}

\begin{definition} [Star Number, \cite{DBLP:journals/jmlr/HannekeY15}]
    The star number $\mathfrak{s}$ of a function class $\mathcal{H}$ is the largest integer $s$ such that there exist $x_1, \ldots, x_s \in X$ and $h_0, h_1,\ldots, h_s \in \mathcal{H}$ such that for each $i,j \in [s]$:
    \begin{equation} \label{eqn:star-condition}
        h_0(x_j) \neq h_i(x_j) \quad \Longleftrightarrow \quad i = j.
    \end{equation}
    That is, the hypothesis $h_i$ disagrees with $h_0$ only on the $i$th element of $\{x_1,\ldots, x_s\}$. If there is no maximum, then we say that the star number is infinite, $\mathfrak{s} = \infty$.
\end{definition}

\begin{lemma} [Star Number of Decision Stumps] \label{lem:ds-star-number}
    The star number of $\DS_d$ is $2d$.
\end{lemma}

\begin{proof}
    We first show that the star number of $\DS_d$ is at least $2d$. Let $e_1,\ldots, e_d$ be the standard basis vectors and let $\mathbf{1} = e_1+ \dotsm + e_d$. Define the following set with $2d$ instances $\{x_{i,-}, x_{i,+}: i \in[d]\}$,
    \begin{align*} 
        x_{i,-} = \left(-1, \ldots, -1 , -\frac{1}{2} , -1, \ldots, -1\right) \quad\textrm{and}\quad
        x_{i,+} = \bigg(+1, \ldots, +1 , +\frac{1}{2}, +1, \ldots, +1\bigg),
    \end{align*}
    where $x_{i,-} = - \mathbf{1} + \frac{1}{2} e_i$ and $x_{i,+} = \mathbf{1} - \frac{1}{2} e_i$. Define the following collection of hypotheses:
    \[h_{i,-} = f_{i, -3/4} \qquad \textrm{and}\qquad h_{i,+} = f_{i, +3/4},\]
    in addition to the base hypothesis $h_{0} = f_{1,0}$. By construction, each hypothesis assigns $x_{i,-}$ the label 0 except for $h_{x,-}$, which assigns it 1. The reverse is true for $x_{i,+}$, where all hypotheses except for $h_{i,+}$ assigns it the label 1. Thus, \cref{eqn:star-condition} holds, so that $\mathfrak{s}(\DS_d) \geq 2d$.

    To show an upper bound on the star number, let $\{x_1,\ldots, x_s\} \in \RR^d$ and $h_0,h_1,\ldots, h_s \in \DS_d$ be any collection of instances and hypotheses satisfying \cref{eqn:star-condition}. Suppose that $s > 2d$. Then, by the pigeonhole principle, there are at least three hypotheses that threshold on the same coordinate. Without loss of generality, suppose that they are the first three hypotheses $h_1, h_2, h_3$, and that they also threshold on the first coordinate:
    \[h_1 = f_{1,t_1} \qquad h_2 = f_{1, t_2}\qquad h_3 = f_{1,t_3},\]
    where $t_1 < t_2 < t_3$. By assumption, we have that $h_2$ disagrees with both $h_1$ and $h_3$ on $x_2$, since:
    \[h_2(x_2) \ne h_0(x_2) = h_1(x_2) = h_3(x_2).\]
    However, $h_1$ and $h_2$ disagree only on points where the first coordinate is between $[t_1, t_2)$, while $h_2$ and $h_3$ disagree only on points where the first coordinate is between $[t_2, t_3)$. Therefore, there is no point $x_2$ that can satisfy \cref{eqn:star-condition}. This implies the upper bound $\mathfrak{s}(\DS_d) \leq 2d$.
\end{proof}

\end{document}